\newif\ifdraft
\drafttrue 
\draftfalse
\ifdraft
\documentclass[11pt,onecolumn,draftcls]{IEEEtran} %
\else
\documentclass[]{IEEEtran} %
\fi
\usepackage{multirow}
\usepackage{amsfonts,amssymb,amsmath,amsthm,mathrsfs}
\usepackage{algorithmic,algorithm}
\usepackage{xcolor,soul}
\usepackage{ifpdf}
\ifpdf
\usepackage[]{graphicx}
\else
\usepackage[draft]{graphicx}
\fi
\usepackage{array}
\usepackage{caption}
\usepackage{subcaption}
\usepackage{url,cite}
\usepackage{breakurl}
\usepackage{changepage}

\newtheorem{definition}{Definition}

\newtheorem{theorem}{Theorem}

\def\bbA{\mathbb{A}}
\def\bbG{\mathbb{G}}
\def\bbZ{\mathbb{Z}}
\def\calA{\mathcal{A}}

\def\calF{\mathcal{F}}

\def\calL{\mathcal{L}}

\def\calT{\mathcal{T}}

\def\inr{\in_r}
\def\scrA{\mathscr{A}}

\def\enc{\mathsf{Enc}}
\def\dec{\mathsf{Dec}}
\def\Setup{\mathsf{Setup}}
\def\KeyGen{\mathsf{KeyGen}}
\def\Encrypt{\mathsf{Encrypt}}
\def\Decrypt{\mathsf{Decrypt}}
\def\DecryptNode{\mathsf{DecryptNode}}

\usepackage[normalem]{ulem}
\definecolor{revision}{rgb}{1,0,0}

\usepackage[bordercolor=white,backgroundcolor=gray!30,linecolor=black,colorinlistoftodos]{todonotes}

\begin{document}
\title{P-MOD: Secure Privilege-Based Multilevel\\ 	
Organizational Data-Sharing in Cloud Computing}
\author{Ehab Zaghloul, Kai Zhou and Jian Ren%
\thanks{The authors are with the Department of Electrical and Computer Engineering, Michigan State University, East Lansing, MI 48824-1226, Email: \{ebz, zhoukai, renjian\}@msu.edu}}

\maketitle

\begin{abstract}
Cloud computing has changed the way enterprises store, access and share data.
Data is constantly being uploaded to the cloud and shared within an organization built on a hierarchy of many different individuals that are given certain data access privileges. With more data storage needs turning over to the cloud, 
finding a secure and efficient data access structure has become a major research issue.
With different access privileges, individuals with more privileges (at higher levels of the hierarchy) are granted access to more sensitive data than those with fewer privileges (at lower levels of the hierarchy). 
In this paper, a Privilege-based Multilevel Organizational Data-sharing scheme~(P-MOD) is proposed that incorporates a privilege-based access structure into an attribute-based encryption mechanism to handle these concerns. 
Each level of the privilege-based access structure is affiliated with an access policy that is uniquely defined by specific attributes. 
Data is then encrypted under each access policy at every level to grant access to specific data users based on their data access privileges. 
An individual ranked at a certain level can decrypt the ciphertext (at that specific level) if and only if that individual owns a correct set of attributes that can satisfy the access policy of that level. 
The user may also decrypt the ciphertexts at the lower levels with respect to the user's level. 
Security analysis shows that P-MOD is secure against adaptively chosen plaintext attack assuming the DBDH assumption holds.
The comprehensive performance analysis demonstrates that P-MOD is more efficient in computational complexity and storage space than the existing schemes in secure data sharing within an organization.
\end{abstract}


\section{Introduction}

It was estimated that data breaches cost the United States' healthcare industry approximately \$6.2 billion in 2016 alone~\cite{ponemon2016}. 
To mitigate financial loss and implications on the reputation associated with data breaches, large multilevel organizations, such as healthcare networks, government agencies, banking institutions,  commercial enterprises and etc., began allocating resources into data security research to develop and improve accessibility and storage of highly sensitive data.

One major way that large enterprises are adapting to increased sensitive data management is the utilization of the cloud environment.
It was reported that more than half of all U.S. businesses have turned over to the cloud for their business data management needs \cite{forbes}.
The on-demand cloud access and data sharing can greatly reduce data management cost, storage flexibility, and capacity \cite{mell2011nist}.
However, data owners have deep concerns when sharing data on the cloud due to security issues. Once uploaded and shared, the data owner inevitably loses control over the data, opening the door to unauthorized data access.

A critical issue for data owners is how to efficiently and securely grant privilege level-based access rights to a set of data. 
Data owners are becoming more interested in selectively sharing information with data users based on different levels of granted privileges. 
The desire to grant level-based access results in higher computational complexity and complicates the methods in which data is shared on the cloud. 
Research in this field focuses on finding enhanced schemes that can securely, efficiently and intelligently share data on the cloud among users according to granted access levels.

Based on a study conducted by the National Institute of Standards and Technology~(NIST), Role-Based Access Control~(RBAC) models are the most widely used to share data in hierarchical enterprises of 500 or more individuals \cite{o20102010}. 
RBAC models aim to restrict system access to authorized users as they provide access control mechanisms. 
The access control mechanisms are based on predefined and fixed roles making the models identity-centric. 
Each individual within the organization is assigned to a role that defines the user's privileges. 
However, the limitations of this model are evident when presented with a large complex matrix of data users in an organization. The foundation of RBAC is based on abstract choices for roles. 
This would require a continuously increasing number of RBAC roles to properly encapsulate the privileges assigned to each user of the system. 
Managing a substantial number of rules can become a resource-intensive task, referred to as \textit{role explosion} \cite{elliott2010role}.

To better comprehend the importance of this study, consider the scenario where patients share their Public Health Records (PHR) on the cloud to be accessed by health providers and administrators of a hospital.
In most cases, the patient wishes to grant the physician access to most parts of the PHR (including its most sensitive parts, e.g. medical history) while granting an administrator access to limited parts that are less sensitive (e.g. date of birth). 
In order to achieve that, the patient needs to define a hierarchy of data access privileges ranking various types of hospital employees. 
Next, the patient needs to clarify the privileges at each level to define what each data user can access. 
It is important to realize that each patient may wish to encrypt his/her PHR differently. 
For example, the patient might grant access to the most sensitive parts of his/her PHR to only specific physicians while denying others. 
This allows the patient full control in defining the hierarchy, which is not fixed or predefined by the hospital.

In this paper, a Privilege-based Multilevel Organizational Data-sharing scheme~(P-MOD) is proposed to solve the problems of sharing data within organizations with complex hierarchies. 
First, the scheme proposes to partition a data file into multiple parts of different sensitivity. 
Next, each part is symmetrically encrypted. The keys used in encryption serve as the actual data shared to users. The scheme then proposes an access structure that ranks data users of an organization into different privileged levels. 
Each level is associated with an access tree that defines the privileges pertaining to data users at each specific level.
Each part of the data file is then encrypted once in a hierarchical manner to grant different access rights to the users based on their level within the hierarchy. 
The process of encryption and decryption is based on an Attribute Based Encryption (ABE) architecture that can achieve fine-granularity when assigning privileges.

The rest of this paper is organized as follows. 
In Section \ref{related}, the related work is reviewed. 
In Section \ref{sec:prelim}, preliminaries are introduced that summarize key concepts used in this research. 
Next, in Section \ref{sec:problem}, the problem formulation is described outlining the system model and design goals.
In Section \ref{sec:scheme}, the proposed scheme, P-MOD is introduced. 
Following that, in Section \ref{sec:security} we formally prove the security of P-MOD based on the hardness of the Decisional Bilinear Diffie-Hellman (DBDH) problem.
In Section \ref{sec:performance}, performance analysis of P-MOD is conducted and compared with two other schemes \cite{wang2016efficient, bethencourt2007ciphertext}. 
Finally, in Section \ref{sec:conclusion}, a conclusion is drawn to summarize the work done in this research.

\section{Related Work \label{related}}

Fuzzy Identity-Base Encryption (Fuzzy IBE) was introduced in \cite{sahai2005fuzzy} to handle data sharing on the cloud in a flexible approach using encryption. 
The ciphertext is shared on the cloud to restrict access to authorized users. 
In order for an authorized individual to obtain the data, the user must request a private key from a key-issuer to decrypt the encrypted data. 
Fuzzy IBE is a specific type of function encryption \cite{boneh2011functional} in which both the data user's private key and ciphertext are affiliated with attributes.
Attributes are descriptive pieces of information that can be assigned to any user or object. 
Since attributes can be any variable, they provide more flexibility when granting data access. 
The scheme enables a set of descriptive attributes to be associated with a data user's private key and the ciphertext shared on the cloud. 
If the data user's private key incorporates the minimum threshold requirement of attributes that match those integrated within the ciphertext, the data user can decrypt it.
Although this scheme allows complex systems to be easily defined using attributes, it becomes less efficient when used to express large systems or when the number of attributes increases. 

Attribute-Based Encryption (ABE) schemes later emerged to provide more versatility when sharing data. 
These schemes integrate two types of constructs: attributes and access policies. Access policies are statements that join attributes to express which users of the system are granted access and which users are denied. 
ABE schemes were introduced via two different approaches: Key-Policy Attribute-Based Encryption (KP-ABE) and Ciphertext Policy Attribute-Based Encryption (CP-ABE). 
KP-ABE was initially introduced in~\cite{goyal2006attribute}. 
In KP-ABE, each ciphertext is labeled with a set of descriptive attributes, while each private key is integrated with an access policy. 
For authorized data users to decrypt the ciphertext, they must first obtain a private key from the key-issuer to use in decryption. 
The key-issuer integrates the access policy into the keys generated. 
Data users can successfully decrypt a ciphertext if the set of descriptive attributes associated with the ciphertext satisfies the access policy integrated within their private keys. KP-ABE can achieve fine-grained access control and is more flexible than Fuzzy IBE. 
However, the data owner must trust the key-issuer to only issue private keys to data users granted the privilege of access. 
This is a limitation since the data owner ultimately forfeits control over which data users are granted access.

CP-ABE is another approach that was later proposed in \cite{bethencourt2007ciphertext}. 
It is considered to be conceptually similar to Role-Based Access Control (RBAC) \cite{ferraiolo2009role}. 
However, CP-ABE gives the data owner control over which data user is able to decrypt certain ciphertexts.
This is due to the access structure being integrated by the data owner into the ciphertext during encryption. 
It allows the private key generated by the key-issuer to only contain the set of attributes possessed by the data user. 
Several enhanced CP-ABE schemes \cite{cheung2007provably, lewko2010fully, denisow2015dynamic, guo2014cp} were later introduced that can provide higher flexibility and better efficiency.

Most attribute-based encryption schemes such as Fuzzy IBE, KP-ABE, and CP-ABE serve as a better solution when data users are not ranked into a hierarchy and each is independent of one another (i.e. no relationships). 
However, they share a common limitation of high computational complexity in the case of large multilevel organizations.
These schemes require a single data file to be encrypted with a large number of attributes (from different levels) to grant them access to it.

Hierarchical Attribute-Based Encryption (HABE) that combines the Hierarchical Identity-Based Encryption (HIBE) \cite{gentry2002hierarchical} scheme and CP-ABE  was later introduced in \cite{wang2010hierarchical}. 
HABE is able to achieve fine-grained access control in a hierarchical organization. 
It consists of a root master that generates and distributes parameters and keys, multiple domain masters that delegate keys to domain masters at the following levels, and numerous users. 
In this scheme, keys are generated in the same hierarchical key generation approach as the HIBE scheme.
To express an access policy, HABE uses a disjunctive normal form where all attributes are administered from the same domain authority into one conjunctive clause.
This scheme becomes unsuitable for practical implementation when replicas of the same attributes are administered by other domain authorities.
Synchronizing attribute administration might become a challenging issue with complex organizations that have multiple domain authorities.
Examples of other hierarchical schemes were introduced in \cite{bobba2009attribute, wan2012hasbe}.

File Hierarchy Ciphertext Policy Attribute-Based Encryption (FH-CP-ABE) scheme was introduced in \cite{wang2016efficient}. 
FH-CP-ABE proposes a leveled access structure to manage a hierarchical organization that shares data of various sensitivity. 
A single access structure was proposed that represents both the hierarchy and the access policies of an organization. 
This access structure consists of a root node, transport nodes, and leaf nodes. The root node and transport nodes are in the form of gates (i.e. AND or OR). 
The leaf nodes represent attributes that are possessed by data users. 
Based on the possession of certain attributes, each data user is mapped into specific transport nodes (certain levels within the hierarchy) based on the access structure that the user satisfies. 
If the data user satisfies a full branch of the access structure, then the data user is ranked at the root node (highest level within the hierarchy). 
Data users ranked at the highest level (root node) can decrypt a ciphertext of highest sensitivity and any other ciphertext with less sensitivity in the lower levels of the hierarchy. 
The nodes ranked in the lower levels (transport nodes) can not decrypt any ciphertexts in the levels above. 
The main advantage of this scheme is that it provides leveled access structures which are integrated into a single access structure. 
As a result, storage space is saved as only one copy of the ciphertext is needed to be shared on the cloud for all data users. 
However, since this scheme uses a single access structure to represent the full hierarchy, the higher levels are forced to accommodate attributes of all the levels below. 
As the number of levels increases in the hierarchy, the number of attributes grows exponentially making this scheme infeasible on a large scale. 
The authors in this work also propose a simplified and reduced access structure to reduce the computational complexity. 
They achieve this by removing all branches of the single access structure while keeping one full branch. 
The full branch consists of the root node, a set of transport nodes (one for each level), and the leaf nodes (attributes). 
They claim that not all nodes in the hierarchy carry information and therefore could be removed. 
However, this claim is only applicable to the case where the highest transport node of each branch to be removed is an OR gate. 
This is the least complicated scenario. 
In the case where the highest transport node of each branch consists of an AND gate, this solution is not applicable. 
Removing branches that consist of AND gates would change the access policies defined. 
In real applications, relationships within an organization are often built in a cross-functional matrix, making this a complicated solution when assigning privileges.  

\section{Preliminaries} \label{sec:prelim}

This section introduces the preliminaries. 

\subsection{Cryptographic Hash Function}
A cryptographic hash function $h$ is a mathematical algorithm that maps data of arbitrary size to a bit string of fixed size. It is cryptographically secure if it satisfies the following requirements:

\begin{itemize}
\item \emph{Preimage-Resistance}: It should be computationally infeasible to find any input for any pre-specified output which hashes to that output, i.e. for any given $y$, it should be computationally infeasible to find an $x$ such that $h(x)=y$.
	
\item \emph{Week Collision Resistance}: For any given $x$, it should be computationally infeasible to find $x'\ne x$ such that $h(x')=h(x)$ \cite{naor1989universal}.
	
\item \emph{Strong Collision-Resistance}: It should be computationally infeasible to find any two distinct inputs $x$ and $x'$, such that $h(x) = h(x')$ \cite{bellare1997collision,damgaard1987collision}.
\end{itemize}

\subsection{Bilinear Maps}
Let $\bbG_0$ and $\bbG_1$ be two multiplicative cyclic groups of the same prime order $p$. The generator of $\bbG_0$ is denoted as $g$. A bilinear map from $\bbG_0 \times \bbG_0$ to $\bbG_1$ is a function $e\colon \bbG_0 \times \bbG_0 \rightarrow \bbG_1$ that satisfies the following properties:
\begin{itemize}
\item \emph{Bilinearity}: $e(g^a,g^b) = e(g,g)^{ab}$ for any $a,b\in\mathbb{Z}_p$.
\item \emph{Symmetry}: $e(g^a,g^b) = e(g,g)^{ab} = e(g^b,g^a)$ for any $a,b\in\mathbb{Z}_p$.
\item \emph{Non-degeneracy}: $e(g,g)\neq 1$.
\item \emph{Computability}: $e(g,g)$ is an efficiently computable algorithm.
\end{itemize}

\subsection{Decisional Bilinear Diffie-Hellman (DBDH) Assumption}

The DBDH assumption \cite{waters2011ciphertext} is a computational hardness assumption and is defined as follows:

Let $\bbG_0$ be a group of prime order $p$, $g$ be a generator, and $a,b,c\in\mathbb{Z}_p$ be chosen at random. 	

It is infeasible for the adversary to distinguish between any given $(g,g^a,g^b,g^c,e(g,g)^{abc})$ and $(g,g^a,g^b,g^c,R)$, where $R\inr\bbG_1$ is a random element and $\inr$ denotes a random selection.		An algorithm $\calA$ that  outputs a guess $z\in\{0,1 \}$, has advantage $\varepsilon$ in solving the DBDH problem in $\bbG_0$ if:
\begin{equation}
 \begin{split}
  &\left|Pr[\calA(g,g^a,g^b,g^c,T=e(g,g)^{abc})=0] \right. \\
  &\left. - Pr[\calA(g,g^a,g^b,g^c,T=R)=0]\right|\geq\varepsilon. 
 \end{split}
\end{equation}

The DBDH assumption holds if no polynomial algorithm has a non-negligible advantage in solving the DBDH problem.

\subsection{Access Structure}

An access structure \cite{beimel1996secure} represents access policies for a set of individuals interested in gaining individual access to a secret. The access structure defines sets of attributes that can be possessed by a single individual to allow access to the secret. It is defined as follows:

Let $\{P_1,\dots,P_n\}$ be a set of parties. A set of parties that can reconstruct the secret is defined as a collection. The collection is monotone meaning that, if $\scrA\subseteq 2^{\{P_1,\dots,P_n \}}$ then $\forall B\in\scrA$ and $B \subseteq C$ implies $C\in\scrA$. An access structure is a monotone collection $\scrA$ of non-empty subsets $\{P_1,\dots,P_n\}$, i.e., $\scrA \subseteq 2^{\{P_1,\dots,P_n \}} \setminus\emptyset$. The sets in $\scrA$ are called the authorized sets and the sets not in $\scrA$ are called the unauthorized sets.

In this paper, we use parties to represent the attributes. This means that an access structure $\scrA$ may consist of both authorized and unauthorized sets of attributes.

\subsection{Leveled Access Tree $\calT_i$}

An access tree $\calT_i$ at level $\calL_i$ represents an access structure that determines whether a data user can decrypt the ciphertext $CT_i$ at that level or not. A $\calT_i$ may consist of multiple nodes. We use $x_l^i \in \calT_i$ to represent the $l^{th}$ node of $\calT_i$. The non-leaf nodes of $\calT_i$ are in the form of threshold gates, while the leaf nodes represent possible attribute values possessed by data users.

For every node $x_l^i \in \calT_i$, a threshold value $k_{x_l^i}$ is assigned. A node in the form of an $\mathsf{AND}$ gate is associated with a threshold value $k_{x_l^i}=num_{x_l^i}$, where $num_{x_l^i}$ represents the number of children of node $x_l^i$. A node in the form of an $\mathrm{OR}$ gate or any leaf node representing an attribute is associated with a threshold value $k_{x_l^i}=1$. 

The root node $x_1^i$ of each $\calT_i$ carries a secret $sk_i$. The data user that possesses the correct set of attributes can satisfy $\calT_i$ and obtain $sk_i$.

\section{Problem Formulation} \label{sec:problem} 

The data owner (DO) owns a data file $\calF$ and wishes to share selective parts of it within an organization O. 
The $\calF$ is assumed to consist of either a single record or multiple records, where each record is composed of one or more record attributes.
The O consists of a large set of data users (DU), where DU = $\{DU_1,DU_2,\dots,DU_m\}$ and $m$ represents the number of data users. 
Data users in the set DU are ranked into a hierarchy H to define their privileges and which parts of $\calF$ they can gain access to. 
The H is not fixed nor predefined. It is defined by the DO as he/she encrypts his/her $\calF$. 
In other words, the DO ranks the set DU in the O based on his/her personal needs.
The DO defines the H based on an infinite pool of attributes. 
This means that the DO can rank the set DU according to his/her possession of specific attributes sets. 
It is important to realize that the DO has the complete freedom to define the H with any set of attributes without personally knowing any of the data users in set DU or what they possess in terms of attribute sets.
The H consists of multiple levels $\{\calL_1,\calL_2,\dots,\calL_k\}$, where $1 \leq k \leq \infty$. $\calL_1$ represents the highest rank while $\calL_k$ represents the lowest rank.
At each level $\calL_i$ of the H, the DO defines a leveled access tree $\calT_i$, where $1 \leq i \leq k$.
Each $DU_j$ possesses a set of attributes $\bbA_j = \{A_{j,1},A_{j,2}\dots,A_{j,n}\}$ that ranks them within a certain level of H.
The $DU_j$ uses his/her set $\bbA_j$ to request a unique private key $SK_j$ from the key-issuer to use for decryption.

The challenge is to provide a secure and efficient scheme for the DO to share different segments of $\calF$ among a set DU, where each user is ranked differently in H.
Each segment of $\calF$ has a sensitivity value that defines which users can gain access to it based on their rank within H.
A user at any given level can gain access to only the segment at his/her level and all the segments from the respective levels below.

\addtolength{\tabcolsep}{-3pt}
\begin{table}[t]
\centering
\caption{Notations summary} \label{tab:notations}
\begin{tabular}{|c|m{2.9in}|} 
\hline
 Symbol & Definition\\ \hline
 DO     &  A data owner\\  \hline
 O    &  An organization that consists of multiple data users\\ \hline  
 H  &  The hierarchical layout of the O\\  \hline
 $\calL_i$ & The $i^{th}$ level within H where $1\leq i\leq k$\\ \hline
 $\calT_i$ & The $i^{th}$ access tree at $\calL_i$ where $1\leq i\leq k$\\  \hline
 $F_i$  &  The $i^{th}$ data file part where $1 \leq i \leq k$\\  \hline
 $sk_i$ &  The $i^{th}$ symmetric key used to encrypt $F_i$ where $1 \leq i \leq k$\\  \hline  
 $EF_i$ &  The $i^{th}$ encrypted $F_i$ under $sk_i$ where $1 \leq i \leq k$\\  \hline  
 $CT_i$ &  The $i^{th}$ ciphertext (encrypt $sk_i$ under $\calT_i$) at $\calL_i$ where $1 \leq i \leq k$\\  \hline   
 $DU_j$ &  The $j^{th}$ data user in set DU where $1 \leq j \leq m$\\ \hline
 $SK_j$ &  The $j^{th}$ data user's private key where $1 \leq j \leq m$\\ \hline 
 $A_{j,u}$ & The $u^{th}$ attribute within the $j^{th}$ data user's attribute set $\bbA_j$ where $1 \leq u \leq n$ and $1 \leq j \leq m$ \\ \hline 
 $x_l^i$ &  The $l^{th}$ node of $\calT_i$ where $1 \leq l \leq \infty$ and $1\leq i\leq k$\\  \hline 
 $k_{x^i_{l}}$ &  Threshold value of $x_l^i$ where $1 \leq l \leq \infty$ and $1\leq i\leq k$\\  \hline 
\end{tabular}
\end{table} 

\subsection{System Model}
The general model of privilege-based data sharing among data users of an organization is illustrated in Fig.~\ref{Fig:Basic Model}. 
In the figure, data users ranked at the higher levels (i.e. possess more privileges) within the hierarchy are granted access to more sensitive data than those ranked at lower levels (i.e. possess fewer privileges). The system consists of four main entities: key-issuer, cloud server, data owner, and data user. 

\begin{itemize}
\item Key-issuer: A fully trusted entity that grants private keys to data users in a system after authenticating their privileges.
\item Cloud server: A non-trusted entity used to store ciphertexts. 
\item Data owner: An individual that owns a data file and wishes to share it with multiple data users of an organization selectively based on their data access privileges. 
\item Data user: An individual that is ranked within a hierarchy of an organization and is interested in decrypting ciphertexts on the cloud. 
\end{itemize}

\begin{figure*}
\centering
\includegraphics[scale=0.6]{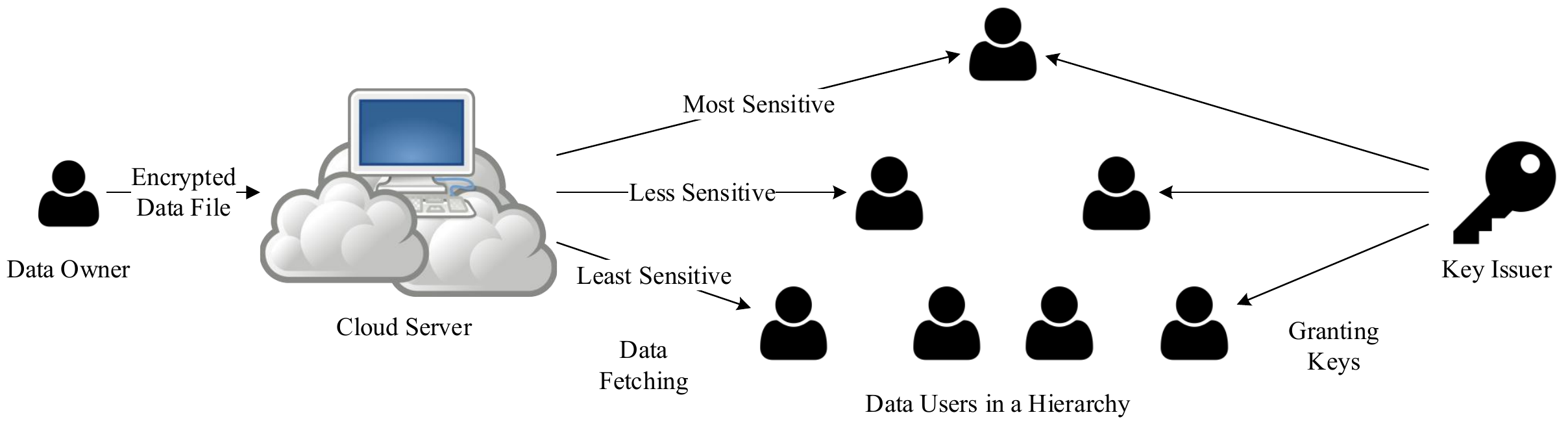}
\caption{General scheme of privilege-based data sharing}
\label{Fig:Basic Model}
\end{figure*}

It is a burden on the data owner to share his/her data file on the cloud as the hierarchy grows (the access privileges increase in number) and/or as the access restrictions become more complex due to an increase in the sensitivity of the data file. 
The data owner wishes to share the data file on the cloud in an efficient manner that is not computationally expensive while minimizing the cloud storage space used.

A trivial solution involves the data owner to use public key encryption. 
Every data user's public key is used to decrypt the part of the data file they are granted access to. 
This ensures that no unprivileged data user will gain access to the data file even if that user is able to download the ciphertext from the cloud server. This solution would require the data owner to encrypt the same part of the data file once for each data user he/she wishes to grant access to. 
On a large scale, public key encryption becomes an inefficient solution due to the increase in the number of encryptions.
It also requires large amounts of storage space making it expensive.

\subsection{Design Goals}
To provide efficient, secure, and privilege-based data sharing to individuals of an organization, we have the following design goals:
\begin{itemize}
\item \emph{Privilege-Based Access}: Data is shared in a hierarchical manner based on user privileges. Data users with more privileges (ranked at the higher levels of the hierarchy) are granted access to more sensitive parts of $\calF$ than those with fewer privileges (ranked at the lower levels of the hierarchy). 

\item \emph{Data Confidentiality}: All parts of $\calF$ are completely protected from the data users that are not privileged (including the cloud) to access the data. Each data user is entitled to access the parts of $\calF$ corresponding to the level they fall in and/or any other parts corresponding to the levels below with respect to the user's level. 

\item \emph{Fine-grained access control}: The data owner has the capability to encrypt any part of $\calF$ using any set of descriptive attributes he/she wishes, limiting access to only authorized data users. The set of descriptive attributes is defined by the data owner at the time of encryption.

\item \emph{Collusion resistant}: Two or more data users at the same/different level can not combine their private keys to gain access to any part of $\calF$ they are not authorized to access independently. 
\end{itemize}

\section{The Proposed P-MOD Scheme}\label{sec:scheme}
This section presents the construction of P-MOD. We assume that file $\calF$ is partitioned into $k$ parts based on data sensitivity. Each part of $\calF$ is independently encrypted and shared among the data users of the system under a privilege-based access structure. 


\subsection{Data File Partitioning and Encryption} \label{subsec:file}
The DO partitions file $\calF$ into a set of $k$ data sections, that is $\calF=\{F_1,F_2,\dots,F_k \}$. 
Each $F_i \in \calF$ is treated as a new file that is associated with a sensitivity value used to assign access rights to the data users based on their privileges. 

The process of partitioning $\calF$ is performed based on the structure of $\calF$. We assume that $\calF$ consists of at least one record, resulting in multiple ways to handle this process. 
Fig.~\ref{Fig:file_part} presents a typical partitioning of $\calF$.

If $\calF$ consists of a single record, then each $F_i\in\calF$ represents one or more record attribute(s) associated with the record, as shown in case (1) in Fig.~\ref{Fig:file_part}. 
However, if $\calF$ consists of multiple records, then DO has flexibility in choosing how to partition it. 
One approach is to handle each record as a whole, where records are clustered into groups of similar sensitivity. 
In this case, each $F_i\in\calF$ represents one or more record(s), as shown case (2) in Fig.~\ref{Fig:file_part}. 
Alternatively, partitioning can be performed over specific record attributes, versus the whole record. 
In this case, $F_i\in\calF$ represents one or more record attribute(s) of the records, as shown case (3) in Fig.~\ref{Fig:file_part}. 
\begin{figure}
\centering
\includegraphics[width=\columnwidth]{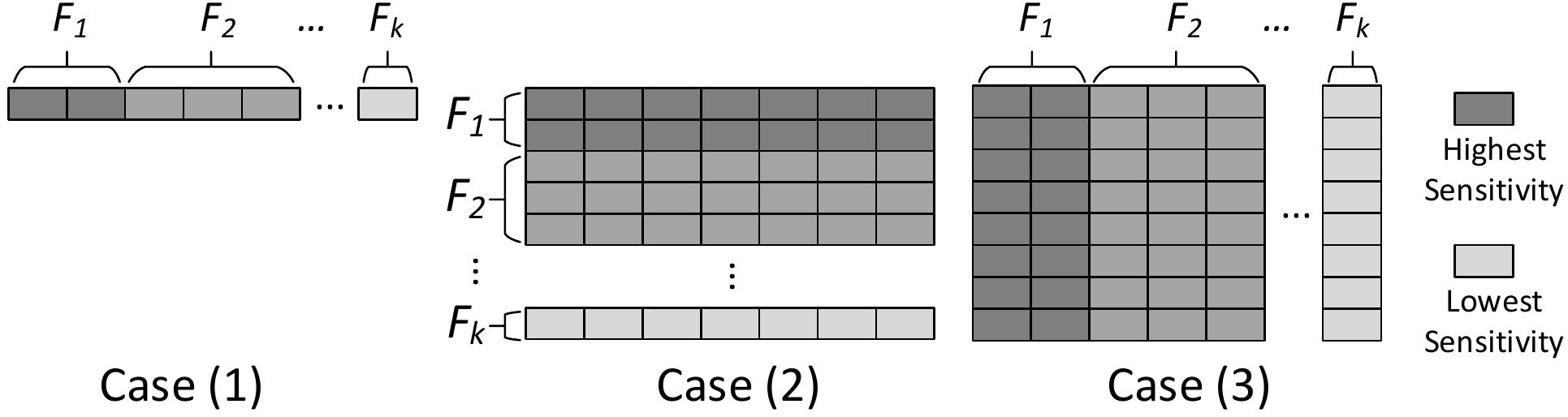}
\caption{File partitioning}
\label{Fig:file_part}
\end{figure}


Each $F_i \in\calF$ is treated as a new data file. Suppose $F_1$ contains the most sensitive information of $\calF$ that can only be accessed by a data user at the highest level $\calL_1$ and $F_k$ contains the least sensitive information of $\calF$ that can be accessed by all data users at any level of the hierarchy. 

Before the DO uploads $\{F_1,F_2,\dots,F_k \}$ to the cloud, each $F_i\in\calF$ is encrypted separately using a symmetric encryption algorithm such as the Advanced Encryption Algorith (AES)~\cite{daemen2013design} with a secret key $sk_i$ to produce an encrypted file $EF_i$. 
\begin{equation}
EF_i = \enc_{sk_i}(F_i).
\end{equation}

For key selection, the DO randomly selects $sk_1$. The remaining symmetric keys $\{sk_2,\dots,sk_k\}$ are derived from $sk_1$ using a public one-way cryptographic hash function $h$, that is
\begin{equation} \label{keys}
sk_{i+1} = h(sk_i).
\end{equation}

After $sk_i$ is used to symmetrically encrypt $F_i$, it is encrypted under P-MOD that will be discussed in Section \ref{Sec:P-MOD-Construction}, to be accessed only by the data users that have been granted the privilege of access. 

The privileged data users that are successful in obtaining $sk_i$ at level $i$ can derive $\{sk_{i+1},\dots,sk_k\}$ using equation (\ref{keys}). However, given the properties of hash function $h$, $sk_i$ cannot be used to derive any of the private keys $\{sk_1,\dots,sk_{i-1}\}$. In other words, a data user can only obtain private keys below his/her level, but cannot obtain any private keys from the levels above.

\subsection{Privilege-Based Access Structure}
The privileged-based access structure divides an O into $k$ levels of privileges, $\{\calL_1,\calL_2,\dots,\calL_k \}$. 
The DO defines a $\calT_i$ at each corresponding $\calL_i$. Each $\calT_i$ is associated with the appropriate leaf nodes (attributes nodes) that define the privilege level. 
Each user within the set DU of O falls into a specific level based on the attributes that the data user possesses. 
The data user that possesses the correct set of attributes which can satisfies $\calT_i$ at $\calL_i$ is able to gain access to parts $\{F_i,F_{i+1},\dots,F_k \}$.

Fig.~\ref{Fig:Model} represents the general privilege-based access structure of  P-MOD. 
At each $\calL_i$, a $\calT_i$ may consist of non-leaf nodes and leaf nodes. 
The non-leaf nodes are in the form of threshold gates, represented as `G', and the leaf nodes are in the form of attributes, represented as `A'. 
The figure represents a privilege-based multilevel access structure. 
It is important to note that each $\calT_i$ could be constructed from any number and layout of nodes based on the privileges at each level.
\begin{figure}
\centering
\includegraphics[width=.95\columnwidth]{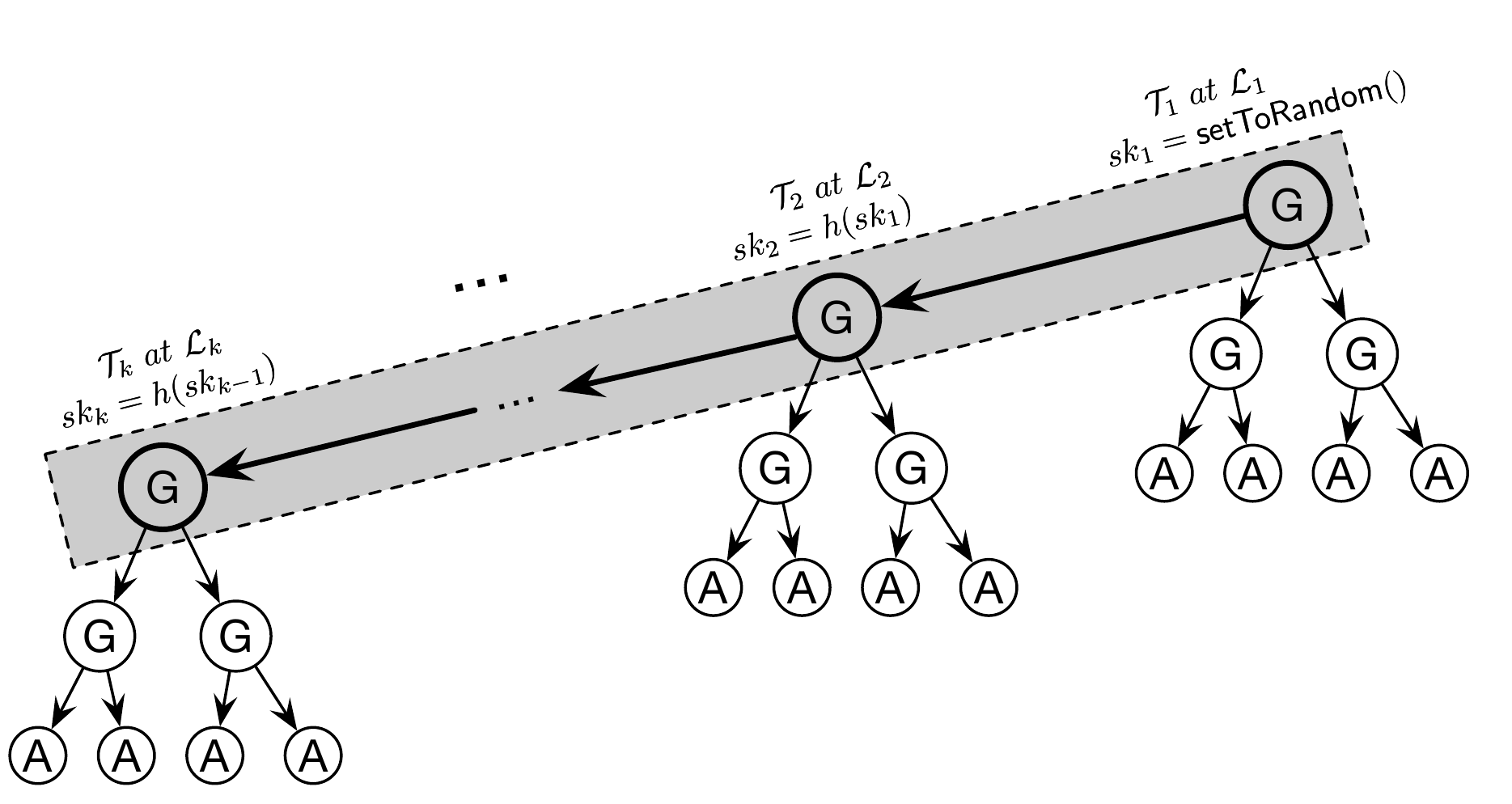}
\caption{Privilege-based multilevel access structure}
\label{Fig:Model}
\end{figure}

The data user is successful in obtaining $F_i$ if and only if the user possesses a correct set of attributes that satisfy $\calT_i$ or any of the access trees $\{ \calT_1,\dots,\calT_{i-1} \}$ corresponding to the higher levels $\{ \calL_1,\dots,\calL_{i-1} \}$ in order to obtain any of $\{sk_1,\dots,sk_i \}$.

\subsection{The Proposed P-MOD Construction} \label{Sec:P-MOD-Construction}
The scheme is based on the construction presented in \cite{bethencourt2007ciphertext} and formally divides the process into four main functions: $\Setup, \KeyGen,\Encrypt, \Decrypt$.

\paragraph{$\Setup(1^\kappa)$}
This is a probabilistic function carried out by the key-issuer. The $\Setup$ function takes a security parameter $\kappa$ and randomly chooses values $\alpha, \beta\in\bbZ_p$. 
The outputs of this function are public key $PK$ and master key $MK$ defined as:
\begin{eqnarray}
PK &=& \{ \bbG_0, g, B=g^{\beta}, e(g,g)^{\alpha} \},\\
MK &=& \{ \beta, g^{\alpha}\}.
\end{eqnarray}

\paragraph{$\KeyGen(MK,\bbA_j)$}
This is also a probabilistic function carried out by the key-issuer. 
The inputs to this function are $MK$ generated by the $\Setup$ function, and the $j^{th}$ data user's attribute set $\bbA_j$, where $A_{j,u}\in \bbA_j$ represents the $u^{th}$ attribute within the set. 
The $\KeyGen$ function outputs a unique private key $SK_j$ for the data user. 
In order to guarantee a unique $SK_j$, it generates a random value $r_j\in\bbZ_p$ and incorporates it within the private key. 
Based on the number of attributes in the input set $\bbA_j$, the $\KeyGen$ function also generates a random value $r_{j,u}\inr \bbZ_p$ for each attribute within the set. The $SK_j$ is defined as:
\begin{equation} \label{privatkey}
\begin{aligned}
SK_j & = \big(D_j = g^{(\alpha + r_j)/ \beta}, \\
& \{ D_{j,u} = g^{r_j} \cdot h(u)^{r_{j,u}}, D_{j,u}' = g^{r_{j,u}}\,|\, \forall A_{j,u} \in \bbA_j\}\big)
\end{aligned}       
\end{equation}

The purpose of the randomly selected $r_j$ is to ensure that each $SK_j$ is unique and the attribute components within the $SK_j$ are associated. 
It should be infeasible for data users to collude by combining components of their private keys ($D_{j,u}$ and $D_{j,u}'$) to decrypt data beyond their individual access rights.
In other words, the attribute components from different private keys cannot be combined to access unauthorized data. 

\paragraph{$\Encrypt(PK$,$sk_i$,$\calT_i)$} 
This is a probabilistic function carried out by the DO. The inputs to this function are $PK$, the public key generated by the $\Setup$ function, the symmetric key $sk_i$ derived in equation (\ref{keys}), and the access tree $\calT_i$ that defines the authorized set of attributes at $\calL_i$. The output of this function is the ciphertext $CT_i$. 

For $sk=\{sk_1,\dots,sk_k\}$, the $\Encrypt$ function will run $k$ times, once for each $sk_i$. 
At each run, the $\Encrypt$ function chooses a polynomial $q_{x_l^i}$ with degree $d_{x_l^i} = k_{x_l^i} - 1$ for each node $x_l^i \in \calT_i$.  
The process of assigning polynomials to each $x_l^i$ occurs in a top-bottom approach starting from the root node in $\calT_i$.  
The $\Encrypt$ function chooses a secret $s \in \bbZ_p$ and sets the value of $q_{x_1^i}(0) = s$. 
Next, it randomly chooses the remaining points of the polynomial to completely define it. 
For any other node $x_l^i \in \calT_i$, the $\Encrypt$ function sets the value $q_{x_l^i}(0) = q_{\mbox{\scriptsize{parent}}}(\mbox{index}(x_l^i))$, where $q_{\mbox{\scriptsize{parent}}}$ is $x_l^i$'s parent node polynomial. The remaining points of those polynomials are then randomly chosen.

Let $X_i$ be the set of leaf nodes in $\calT_i$. The $CT_i$ at $\calL_i$ is then constructed as:
\begin{equation} \label{cipher}
\begin{aligned}
CT_i & = \big(\calT_i, \tilde{C}_i = sk_i \cdot e(g,g)^{\alpha\cdot s}, C_i = g^{\beta\cdot s}, \\
& \{ C_{x_l^i} = g^{q_{x_l^i}(0)}, C_{x_l^i}' = h(x_l^i)^{q_{x_l^i}(0)} \,|\, \forall x_l^i \in X_i \}\big) .     
\end{aligned}
\end{equation}

\paragraph{$\Decrypt(CT_i,SK_j)$}
This is a deterministic function carried out by the data user. 
The inputs to this function are the ciphertext $CT_i$ generated by the $\Encrypt$ function corresponding to $\calL_i$, and the $j^{th}$ data user's private key $SK_j$. 

The $\Decrypt$ function operates in a recursive manner propagating through the nodes in $\calT_i$ by calling a recursive function defined as $\DecryptNode(CT_i,SK_j,x_l^i)$. 
The inputs to this function are $CT_i$, $SK_j$ and a node $x_l^i$ within $\calT_i$. 
If the attributes incorporated within the data user's private key satisfy the nodes within $\calT_i$, the data user can decrypt $CT_i$. 

$\DecryptNode(CT_i,SK_j,x_l^i)$ performs differently depending on whether $x_l^i$ is a leaf or non-leaf node.
If $x_l^i$ is a leaf node and $\mbox{att}(x_l^i) \notin \bbA_j$, $\DecryptNode(CT_i,SK_j,x_l^i)$ returns $\emptyset$, otherwise, $\mbox{att}(x_l^i) = A_{j,u} \in \bbA_j$ and $\DecryptNode(CT_i,SK_j,x_l^i)$ is defined as:
\arraycolsep=2pt 
\begin{eqnarray} 
\mathsf{DecryptNode}(CT_i,SK_j,x_l^i) &=& \frac{e(D_{j,u},C_{x_l^i})}{e(D_{j,u}',C_{x_l^i}')}\nonumber \\
& =& \frac{e(g^{r_j}\cdot h(u)^{r_{j,u}},g^{q_{x_l^i}(0)})}{e(g^{r_{j,u}},h(x_l^i)^{q_{x_l^i}(0)})}\nonumber\\
& =& e(g,g)^{r_j\cdot q_{x_l^i}(0)}. 
\end{eqnarray}
\arraycolsep=2pt 

If $x_l^i$ is a non-leaf node, $\DecryptNode(CT_i,SK_j,x_l^i)$ operates recursively. For each node $z_{l,c}^i$ that is a child of $x_l^i$, $\DecryptNode(CT_i,SK_j,z_{l,c}^i)$ is computed and the output is stored in $F_{z_{l,c}^i}$.

This recursive function is based on Lagrange interpolation.  The Lagrange coefficient $\Delta_{a,\bbA_j}$ for $a \in \bbZ_p$ and the set of attributes $\bbA_j$ is defined as:
\begin{equation}
\Delta_{a,\bbA_j}(x)=\prod_{k\in\bbA_j,k\neq a}\frac{x-k}{a-k}.
\end{equation}

Let $\bbA_{j,x_l^i}$ be an arbitrary $k_{x_l^i}$-sized set of child nodes $z_{l,c}^i$ such that $F_{z_{l,c}^i} \neq \emptyset$. If no such set exists then the function returns $F_{z_{l,c}^i} = \emptyset$. Otherwise, $F_{z_{l,c}^i}$ is computed using Lagrange interpolation as follows:
\begin{equation}
\begin{aligned}
F_{x_l^i} 
    & = \prod_{z_{l,c}^i \in \bbA_{j,x_l^i}} F_{z_{l,c}^i}^{\Delta_{a,\bbA_{j,x_l^i}'}(0)} \\
    & = \prod_{z_{l,c}^i \in \bbA_{j,x_l^i}} \left(e(g,g)^{r_j\cdot q_{z_{l,c}^i}(0)}\right)^{\Delta_{a,\bbA_{j,x_l^i}'}(0)} \\
    & = \prod_{z_{l,c}^i \in \bbA_{j,x_l^i}} \left(e(g,g)^{r_j\cdot q_{\rm{parent}(z_{l,c}^i)}(\rm{index}(z_{l,c}^i))}\right)^{\Delta_{a,\bbA_{j,x_l^i}'}(0)} \\
    & = \prod_{z_{l,c}^i \in \bbA_{j,x_l^i}} e(g,g)^{r_j\cdot q_{x_l^i}(i) \cdot \Delta_{a,\bbA_{j,x_l^i}'}(0)} \\
    & = e(g,g)^{r_j \cdot q_{x_l^i}(0)},
\end{aligned}
\end{equation}
where $a=\rm{index}(z_{l,c}^i)$ and $\bbA_{j,x_l^i}' = \{ \mbox{index}(z_{l,c}^i)\colon z_{l,c}^i \in \bbA_{j,x_l^i} \}$. 

If the attributes in $\bbA_j$ satisfy $\calT_i$, then the following is computed at the root node $x_1^i$.
\begin{equation} \label{eq_A}
\begin{aligned}
R_i & = \mathsf{DecryptNode}(CT_i,SK_j,x_1^i) \\
    & = e(g,g)^{r_j \cdot q_{x_1^i}(0)} \\
    & = e(g,g)^{r_j \cdot s}.
\end{aligned}
\end{equation}

To obtain $sk_i$ from the result derived in equation~(\ref{eq_A}), we compute the following:
\begin{equation} \label{eq_sk}
\frac{\tilde{C}_i}{\frac{e(C_i,D_j)}{R_i}} = \frac{sk_i \cdot e(g,g)^{\alpha\cdot s}}{\frac{e\left(g^{\beta\cdot s},g^{(\alpha + r_j)/ \beta}\right)}{e(g,g)^{r_j\cdot s}}} = sk_i.
\end{equation}

At this point, the data user can simply decrypt $EF_i$ using the $sk_i$ derived from equation (\ref{eq_sk}) to obtain the plaintext $F_i$ as follows:
\begin{equation}   \label{F_dec}
F_i = \dec_{sk_i}(EF_i).
\end{equation}

If the data user is successful in obtaining $sk_i$ and is interested in attaining any information in the lower levels with respect to $\calL_i$, he/she can compute the symmetric keys of the lower level as previously mentioned in equation~(\ref{keys}). All the data files at the lower levels are decrypted using equation (\ref{F_dec}).

\section{Security Analysis \label{sec:security}}

In this section, the security of P-MOD is analyzed. 
It is assumed that a symmetric encryption technique such as AES is used to secure each individual data file $F_i \in\calF$. 
It is also assumed that the process of attribute authentication between a data user and the key-issuer, in order for the data user to obtain a private key, is secure and efficient. 
Therefore, the main focus of this section is to provide a formal proof of security for P-MOD based on the work presented in \cite{waters2011ciphertext}. The security proof provided in this section is based on ciphertext indistinguishability which proves that the adversary is not able to distinguish pairs of ciphertexts. 
A cryptosystem is considered to be secure under this property if the probability of an adversary to identify a data file that has been randomly selected from a two-element data file chosen by the adversary and encrypted does not significantly exceed $\frac12$.
We first present an Indistinguishability under Chosen-Plaintext Attack (IND-CPA) security game. 
Next, based on the IND-CPA security game, a formal proof of security is provided for  P-MOD. 

IND-CPA is a game used to test for security of asymmetric key encryption algorithms. 
In this game, the adversary is modeled as a probabilistic polynomial-time algorithm. 
The algorithms in the game must be completed and the results returned within a polynomial number of time steps. The adversary will choose to be challenged on an encryption under a leveled access tree $\calT^*$. 
The adversary can impersonate any data user and request many private keys $SK_j$. 
However, the game rules require that any attribute set $\bbA_j$ that the adversary claims to possess does not satisfy $\calT^*$.
The security game is divided into the following steps:
\begin{enumerate}
\item \textit{Initialization}: The adversary selects an access tree $\calT^*$ to be challenged against and commits to it (i.e. the adversary will not change it throughout the game).

\item \textit{Setup}: The challenger runs the $\Setup$ function and sends the public key $PK$ to the adversary. 

\item \textit{Phase 1}: The adversary requests multiple private keys ($SK_1,\dots,SK_{q_1}$) corresponding to $q_1$ different sets of attributes ($\bbA_1,\dots,\bbA_{q_1}$).

\item \textit{Challenge}: The adversary submits two equal length data files $F_0$ and $F_1$ to the challenger. The adversary also sends $\calT^*$ such that none of ($SK_1,\dots,SK_{q_1}$) generated from Phase 1 contain correct sets of attributes that satisfy it. The challenger flips a random coin $\mu$ and encrypts $F_{\mu}$ under $\calT^*$. Finally, the challenger sends the ciphertext $CT^*$ generated according to equation (\ref{cipher}) to the adversary.

\item \textit{Phase 2}: Repeat phase 1 with the restriction that none of the newly generated private keys ($SK_{q_1+1},\dots,SK_{q}$) corresponding to the $q - (q_1+1)$ different sets of attributes ($\bbA_{q_1+1},\dots,\bbA_q$) contain correct sets of attributes that satisfy $\calT^*$.

\item \textit{Guess}: The adversary outputs a guess $\mu'$ of $\mu$. The adversary wins the security game if $\mu' = \mu$ and loses otherwise. 
\end{enumerate}

\begin{definition} [Secure against adaptively chosen plaintext attack.]
P-MOD is said to be secure against an adaptively chosen plaintext attack if any polynomial-time adversary has only a negligible advantage in the security game, where the advantage is defined as $Adv = Pr[\mu'=\mu] - \frac{1}{2}$.
\end{definition}

The security of P-MOD is reduced to the hardness of the DBDH problem. By proving that a single $CT_i$ at any $\calL_i$ is secure, the whole system is proved to be secure since all ciphertexts at any level follow the same rules.
\begin{theorem}
P-MOD is secure against adaptively chosen plaintext attack if the DBDH assumption holds.
\end{theorem}
\begin{proof}
Assume there is an adversary that has non-negligible advantage $\varepsilon = \mathsf{Adv}_\calA$. 
We construct a simulator that can distinguish a DBDH element from a random element with advantage $\varepsilon$. 
Let $e\colon\bbG_0 \times \bbG_0 \rightarrow \bbG_1$ be an efficiently computable bilinear map and $\bbG_0$ is of prime order $p$ with generator $g$. 
The DBDH challenger begins by selecting the random parameters: $a,b,c\inr\bbZ_p$. Let $g\in\bbG_0$ be a generator and $T$ is defined as $T=e(g,g)^{abc}$ if $\mu=0$, and $T=R$ otherwise, 
where $\mu\inr\{0,1 \}$ and $R\inr\bbG_1$. The simulator acts as the challenger in the following game:

\begin{enumerate}
\item \textit{Initialization}: The simulator accepts the DBDH challenge requested by the adversary who selects an access structure $\calT^*$.

\item \textit{Setup}: The simulator runs the $\Setup$ function. It chooses a random $\alpha^*\inr\bbZ_p$ and computes the value $\alpha=\alpha^*+ab$. Next, it simulates $e(g,g)^{\alpha} \leftarrow e(g,g)^{\alpha^*+ab}=e(g,g)^{\alpha^*}e(g,g)^{ab}$ and $B=g^{\beta}\leftarrow g^b$, where $b$ represents a simulation of the value $\beta$. Finally, it sends all components of $PK = \{\bbG_0, g, B=g^{\beta}, e(g,g)^{\alpha}\}$ to the adversary.

\item \textit{Phase 1}: In this phase, the adversary requests multiple private keys ($SK_1,\dots,SK_{q_1}$) corresponding to $q_1$ different sets of attributes ($\bbA_1,\dots,\bbA_{q_1}$). After receiving an $SK_j$ query for a given set $\bbA_j$ where $\bbA_j \notin \calT^*$ (i.e. $\forall A_{j,u} \in \bbA_j$ does not satisfy $\calT^*$), the simulator chooses a random $r_j'\inr\bbZ_p$ and defines $r_j=r_j' - b$. Next, it simulates $D_j = g^{(\alpha + r_j)/\beta} \leftarrow g^{\alpha/\beta}g^{r_j/\beta} = g^{(\alpha^* + ab)/b}g^{(r_j'-b)/{b}}$. Then, $\forall A_{j,u} \in \bbA_j$, it selects a random $r_{j,u}\inr\bbZ_p$ and simulates $D_{j,u} = g^{r_j} \cdot h(u)^{r_{j,u}} \leftarrow g^{r_j'-b}h(u)^{r_{j,u}}$ and $D_{j,u}' \leftarrow g^{r_{j,u}}$. Finally, the simulated values of $SK_j = (D_j, \{D_{j,u}, D_{j,u}'\,|\,r_{j,u}\inr\bbZ_p,\ \forall A_{j,u} \in \bbA_j \})$ are sent to the adversary.

\item \textit{Challenge}: The adversary sends two plaintext data files $sk_0$ and $sk_1$ to the simulator who randomly chooses a $\mu\inr\{0,1\}$ by flipping a coin to select one of the files. The simulator then runs the $\Encrypt$ function and derives a ciphertext $CT^*$. It simulates $\tilde{C} = sk_{\mu}\cdot e(g,g)^{\alpha\cdot sk_{\mu}} \leftarrow sk_{\mu}\cdot e(g,g)^{(\alpha^*+ab)c}= sk_\mu\cdot T e(g,g)^{\alpha^* c}$, where $c$ represents a simulation of the value $sk_{\mu}$ and $T=e(g,g)^{abc}$. Next, it simulates $C = g^{\beta\cdot sk_{\mu}} \leftarrow g^{bc}$. Finally, for each attribute $x\in X^*$ (set of leaf nodes in $\calT^*$) it computes $C_x=g^{q_x(0)}$ and $C_x'=h(x)^{q_x(0)}$.  The simulated values of $CT^*=\{\calT^*,\tilde{C},C,\forall x\in X^*\colon C_x, C_x'\}$ are then sent to the adversary.

\item \textit{Phase 2}: Repeat \textit{Phase 1}. 

\item \textit{Guess}: The adversary tries to guess the value $\mu$. If the adversary guesses the correct value, the simulator outputs 0 to indicate that $T=e(g,g)^{abc}$, or 1 to indicate that $T=R$, a random group element in $\bbG_1$.
\end{enumerate}

Given a simulator $\calA$, if $T=e(g,g)^{abc}$, then $CT^*$ is a valid ciphertext, $Adv = \varepsilon$ and
\begin{equation}\label{eq-adv-advantage}
Pr\left[\calA\left(g,g^a,g^b,g^c,T=e(g,g)^{abc}\right)=0\right] = \frac{1}{2} + \varepsilon.
\end{equation}

If $T=R$ then $\tilde{C}$ is nothing more than a random value to the adversary. Therefore, 
\begin{equation}\label{eq-random-proof}
Pr\left[\calA\left(g,g^a,g^b,g^c,T=R\right)=0\right] = \frac{1}{2}.
\end{equation}
From equation (\ref{eq-adv-advantage}) and equation (\ref{eq-random-proof}), we can conclude that 

\begin{equation}
 \begin{split}
  &\left|Pr\left[\calA\left(g,g^a,g^b,g^c,T=e(g,g)^{abc}\right)=0\right] \right. \\
  &\left. - Pr\left[\calA\left(g,g^a,g^b,g^c,T=R\right)=0\right]\right|=\varepsilon. 
 \end{split}
\end{equation}
Therefore, the simulator plays the DBDH game with a non-negligible advantage and the proof is complete.
\end{proof}

\section{Performance Analysis} \label{sec:performance}

In this section, we present a performance analysis for P-MOD. 
We compare it with two existing schemes, CP-ABE~\cite{bethencourt2007ciphertext} and FH-CP-ABE \cite{wang2016efficient}.

\subsection{CP-ABE in a Hierarchical Organization}

CP-ABE \cite{bethencourt2007ciphertext} handles sharing of independent pieces of data based on independent access policies.  
It was not designed to support a privilege-based access structure (i.e. hierarchical organization).
Therefore, to adapt CP-ABE to a privilege-based access structure, the $\Encrypt$ function runs once for each level. 
However, if it were to be used in a hierarchical organization, there would be a trade off between the key management and the complexity of the encryption and decryption processes.
Fig. \ref{Fig:cpabe-hierarchy} shows the two general cases in which CP-ABE is utilized to share data with users in a hierarchical organization.
\begin{figure}
	\centering
	\includegraphics[width=.9\columnwidth]{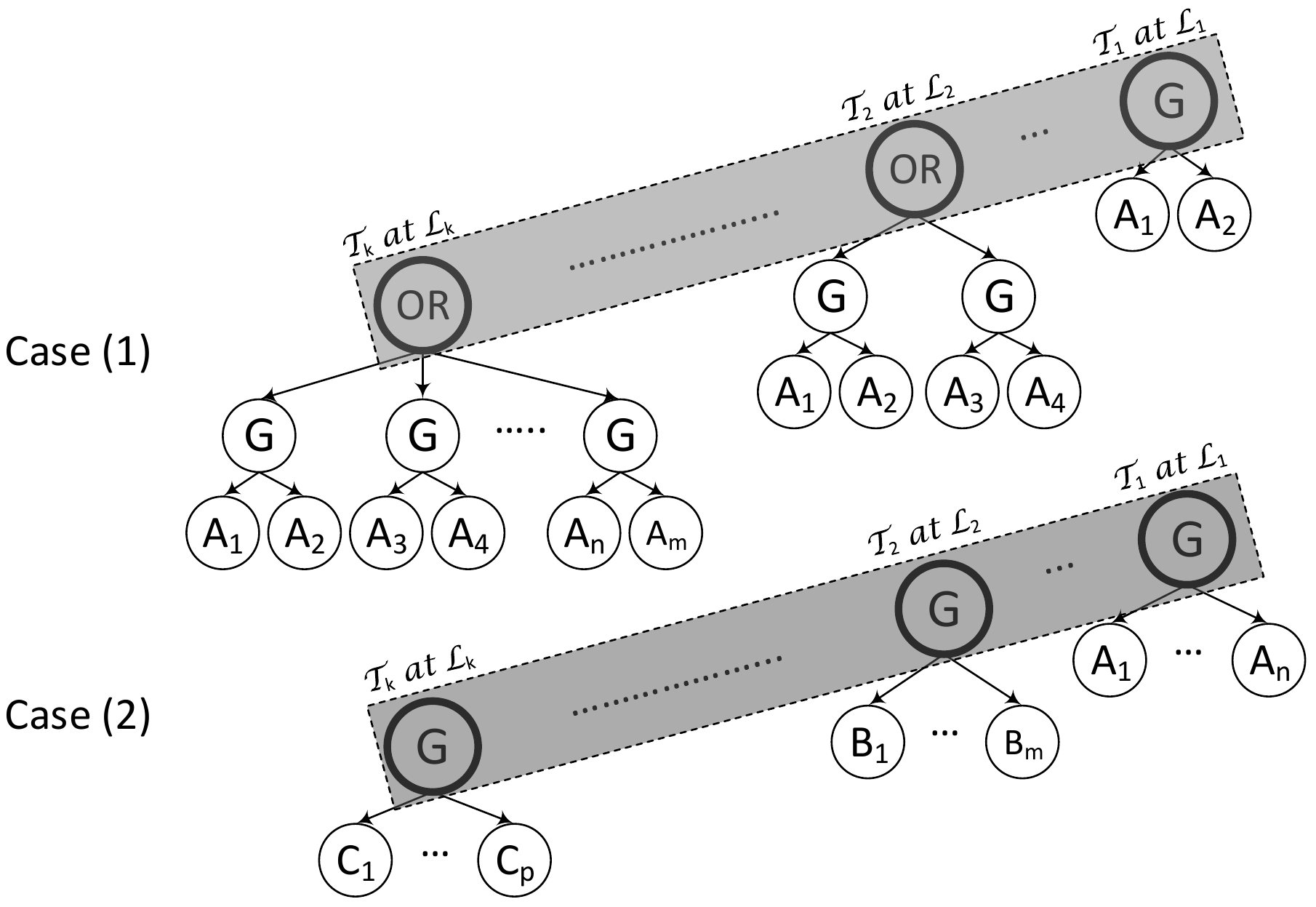}
	\caption{CP-ABE utilized in a hierarchical organization}
	\label{Fig:cpabe-hierarchy}
\end{figure}

In case (1), key management is favored over encryption and decryption complexities. 
The individuals at each level possess attributes that define their privileges and the private keys to encrypt the data files.
The size of each private key is therefore optimized and the key management process becomes less resource-intensive.
However, since levels are independent, attributes must be repeatedly incorporated into each access tree. 
As a result, the sizes of the access trees at lower levels will increase, as shown in case (1) of Fig. 4. 
For complex organizations with a large number of levels,  the access trees will become even larger. 
This results in an increase in encryption and decryption complexities.

On the other hand, case (2) favors minimizing encryption and decryption complexities over key management.
Each data file is encrypted under an access policy with a set of unique attributes at each level without considering privileges and relationships.
This results in simpler access trees at each level of the hierarchy and therefore lower encryption and decryption complexities.
However, to grant access to an individual at a specific level, the key-issuer must generate a private key for that individual that incorporates the attributes at that level and all the levels below.
Complex hierarchies that include a large number of attributes, could result in complicated key management.
As a result, private keys will require incorporating a large number of attributes. 

\subsection{Computational Complexity}
The computational complexity includes both encryption complexity and decryption complexity that are formulated using the number of group operations $f_{\bbG_0}$, $f_{\bbG_1}$ for groups $\bbG_0$, $\bbG_1$ respectively and the number of bilinear mapping operations $e$ involved in the $\Encrypt$ and the $\Decrypt$ functions for each scheme. 

\subsubsection{Encryption Complexity} \label{enc_comp_complexity} 
Encryption complexity represents the number of basic operations involved in generating the ciphertext from the plaintext. It is formulated based on the $\Encrypt$ function which involves group operations $f_{\bbG_0}$ and $f_{\bbG_1}$. 

The number of operations involved in sharing an independent piece of data using CP-ABE is $(2|X|+1)$ and $2$ for $f_{\bbG_0}$ and $f_{\bbG_1}$ respectively, where $|X|$ denotes the number of leaf nodes (attributes) of the access tree $\calT$.
For a hierarchical organization, we present the encryption complexity of case (1) in Fig. \ref{Fig:cpabe-hierarchy} as it involves a relationship between all levels.
In a real life application, case (2) would not satisfy a hierarchical organization as it requires attributes to be shared by all users regardless which level they belong to.
The number of operations involved in the encryption process is formulated as $(2(|X_1|+\cdots+|X_k|)+k)$ and $2k$ for $f_{\bbG_0}$ and $f_{\bbG_1}$ respectively, where $|X_1|,|X_2|,\cdots,|X_k|$ are the number of leaf nodes (attributes) associated with access trees $\calT_1,\calT_2, \cdots,\calT_k$ respectively. 
However, reuse of attributes needed at each level in this scheme increases the computational complexity, making it an overall inefficient solution for hierarchical organizational structures.

FH-CP-ABE \cite{wang2016efficient} uses a single access tree $\calT$ to encrypt all data files to be shared. 
The $\Encrypt$ function only needs to run once for all data files to generate a single ciphertext. 
The number of operations involved are $(2|X|+k)$ and $(2v|\bbA_T| +2k)$ for operations $f_{\bbG_0}$ and $f_{\bbG_1}$ respectively, where $|\bbA_T|$ is the number of transport nodes (levels), and $v$ is the number of children nodes associated with a transport node. 
The single access structure merges access policies across levels into one large policy. 
The combined policy grants users access at different levels. 
This results in large sets $X$, $\bbA_T$ and $v$. 
The encryption complexity in this scheme is dependent on the size of these sets and can grow exponentially depending on how the access tree $\calT$ is constructed.

P-MOD generates multiple ciphertexts (one at each level) of the hierarchy. 
The number of operations involved is $(2(|Y_1|+\cdots+|Y_k|)+k)$ and $2k$ for $f_{\bbG_0}$ and $f_{\bbG_1}$ respectively, where $|Y_1|,|Y_2|,\cdots,|Y_k|$ are the number of leaf nodes (attributes) associated with access trees $\calT_1,\calT_2, \cdots,\calT_k$ respectively.
However, P-MOD leverages the relation between levels of a hierarchy. 
Therefore, it does not require reuse of attributes in levels like CP-ABE.  
In this case, $|Y_1|,|Y_2|,\cdots,|Y_k|$ are all specific to each level minimizing the number of attributes in each access tree.
As a result, $|Y_i| < |X_i|, \forall i \in k$.
Therefore, the encryption complexity of P-MOD outperforms that of CP-ABE. 

When comparing the encryption complexities of P-MOD with FH-CP-ABE, P-MOD minimizes the number of operations involved in contrast to those needed by FH-CP-ABE. 
This is because the encryption process for FH-CP-ABE involves an access tree that contains all the access policies for all the levels in contrast to P-MOD which involves smaller access trees, each one limited to an access policy of a single level. 
As a result, when comparing the encryption complexity of these two schemes, P-MOD is more efficient than FH-CP-ABE.

\subsubsection{Decryption Complexity} \label{dec_comp_complexity}
Decryption complexity represents the number of basic operations involved in decrypting the ciphertext into plaintext. It is formulated based on the $\Decrypt$ function which involves bilinear operations $e$ and group operations $f_{\bbG_1}$. 

For a single $\Decrypt$ run, CP-ABE \cite{bethencourt2007ciphertext} involves $(2|\bbA_j|)$ and $(2|S|+2)$ number of operations $e$ and $f_{\bbG_1}$ respectively, where $|\bbA_j|$ is the number of attributes possessed by the $j^{th}$ data user and $|S|$ is the least number of interior nodes that satisfy $\calT$.

To adapt CP-ABE to a privilege-based access structure, the $\Decrypt$ function is run as many times as the number of ciphertexts a data user wishes to decrypt. 
The number of operations involved are $k(2|\bbA_j|+1)$ and $(2[|S_1|+\dots+|S_k|]+2k)$ for operations $e$ and $f_{\bbG_1}$ respectively, where $|S_1|,|S_2|,\cdots,|S_k|$ are the least number of interior nodes that satisfy the access trees $\calT_1,\calT_2,\cdots,\calT_k$ respectively.

FH-CP-ABE \cite{wang2016efficient} aims to satisfy a certain transport node (level) of the single access tree $\calT$ allowing the data user to decrypt certain encrypted files up to that level. To accomplish this, the $\Decrypt$ function propagates through the access tree nodes recursively to make sure that the user's attributes satisfy the nodes. The number of operations involved in this process are $(2|\bbA_j|+1)$ and $(2|S|+v|\bbA_T|+2k)$ for operations $e$ and $f_{\bbG_1}$ respectively.

Since FH-CP-ABE uses a single access tree $\calT$, the number of transport nodes $|\bbA_T|$ can be large, resulting in higher decryption complexity.
The least number of interior nodes $|S|$ that satisfy $\calT$ can also be large if the construction of the access tree $\calT$ is not optimized. Constructing a single access tree that accommodates a large number of attributes is resource-intensive and could become complicated as the access rules become more sophisticated.

The number of operations involved in the decryption process for P-MOD is similar to a single CP-ABE decryption run. 
P-MOD needs to run the $\Decrypt(CT_i,SK_j)$ function only one time, even if the data user needs to obtain more than one data file.
The $\Decrypt$ function requires $(2|\bbA_j|)$ and $(2|S|+2)$ for operations $e$ and $f_{\bbG_1}$ respectively. 
Once the data user successfully decrypts the ciphertext (obtains the symmetric key at his/her level), the user can derive the remaining lower level keys as described in equation~(\ref{keys}). 
The complexity of the operations involved in deriving the symmetric keys and decrypting ciphertexts at the lower levels are negligible in comparison with the group and bilinear operations involved in running the $\Decrypt$ function, and therefore could be ignored. 

\def\arraystretch{1.2} 
\begin{table*}[t]
\centering
\caption{Computational complexity summary} 
\label{tab:Computational}
\begin{tabular}{c c|c|c|c|}
\cline{3-5} & & \multicolumn{3}{|c|}{\textbf{Number of operations}}\\	
\cline{1-5} \multicolumn{1}{|c|}{\textbf{Complexity}} & \textbf{Operation} & \textbf{CP-ABE} & \textbf{FH-CP-ABE} & \textbf{P-MOD}\\	
\hline \multicolumn{1}{|c|}{\multirow{2}{*}{\textbf{Encryption}}} & \multirow{1}{*}{$f_{G_0}$} & $2(|X_1|+\dots+|X_k|)+k$ & $2|X|+k$ & $2(|Y_1|+\dots+|Y_k|)+k$\\ 
\cline{2-5} \multicolumn{1}{|c|}{} & \multicolumn{1}{|c|}{\multirow{1}{*}{$f_{G_1}$}} & $2k$ & $2v|\bbA_T|+2k$ & $2$\\
\hline \multicolumn{1}{|c|}{\multirow{2}{*}{\textbf{Decryption}}} & \multirow{1}{*}{$e$} & $k(2|\bbA_j|+1)$ & $2|\bbA_j|+1$ & $2|\bbA_j|$\\ 
\cline{2-5} \multicolumn{1}{|c|}{} & \multicolumn{1}{|c|}{\multirow{1}{*}{$f_{G_1}$}} & $2(|S_1|+\dots+|S_k|)+2k$ & $2|S|+v|\bbA_T|+2k$ & $2|S|+2$\\
\cline{1-5}
\end{tabular}
\end{table*}

\def\arraystretch{1.2} 
\begin{table*}[t]
\centering
\caption{Storage efficiency summary} 
\label{tab:Storage}
\begin{tabular}{c c|c|c|c|}
\cline{3-5} & & \multicolumn{3}{|c|}{\textbf{Number of elements}}\\	
\cline{1-5} \multicolumn{1}{|c|}{\textbf{Component}} & \textbf{Bit-length} & \textbf{CP-ABE} & \textbf{FH-CP-AB}E & \textbf{P-MOD}\\	
\hline \multicolumn{1}{|c|}{\multirow{2}{*}{\textbf{Public Key}}} & \multirow{1}{*}{$L_{\bbG_0}$} & $3$ & $3$ & $3$\\ 
\cline{2-5} \multicolumn{1}{|c|}{} & \multicolumn{1}{|c|}{\multirow{1}{*}{$L_{\bbG_1}$}} & $1$ & $1$ & $1$\\
\hline \multicolumn{1}{|c|}{\multirow{2}{*}{\textbf{Master Key}}} & \multirow{1}{*}{$L_{\bbZ_p}$} & $1$ & $1$ & $1$\\ 
\cline{2-5} \multicolumn{1}{|c|}{} & \multicolumn{1}{|c|}{\multirow{1}{*}{$L_{\bbG_0}$}} & $1$ & $1$ & $1$\\
\hline \multicolumn{1}{|c|}{\multirow{1}{*}{\textbf{Secret Key}}} & \multirow{1}{*}{$L_{\bbG_0}$} & $2|\bbA_j|+1$ & $2|\bbA_j|+1$ & $2|\bbA_j|+1$\\ 
\hline \multicolumn{1}{|c|}{\multirow{2}{*}{\textbf{Ciphertext}}} & \multirow{1}{*}{$L_{\bbG_0}$} & $2(|X_1|+\dots+|X_k|)+k$ & $2|X|+k$ & $2(|Y_1|+\dots+|Y_k|)+k$\\ 
\cline{2-5} \multicolumn{1}{|c|}{} & \multicolumn{1}{|c|}{\multirow{1}{*}{$L_{\bbG_1}$}} & $k$ & $v|\bbA_T|+k$ & $k$\\
\cline{1-5}
\end{tabular}
\end{table*}

When comparing the decryption complexities of P-MOD and FH-CP-ABE, P-MOD's decryption complexity depends on $|\bbA_j|$ and $|S|$ while FH-CP-ABE's decryption complexity depends on $|S|$, $v$, $|\bbA_T|$ and $k$. 
The size of the sets in FH-CP-ABE will always be greater than the size of the sets in P-MOD due to the different constructions of access trees in each scheme. In conclusion, P-MOD  outperforms FH-CP-ABE.

\subsection{Storage Efficiency}
The following analysis presents the space needed for keys and ciphertext storage in each scheme. 
The bit-length of the different keys involved and the ciphertexts generated in each scheme are presented. 
The bit-length of a single element in $\bbG_0$, $\bbG_1$ and $\bbZ_p$ are denoted as $L_{\bbG_0}$, $L_{\bbG_1}$ and $L_{\bbZ_p}$ respectively.

\subsubsection{Key Size}
Key size represents the bit-length of a generated key. Three different types of keys are generated for each scheme: a public key, a master key, and user private keys. 

The sizes of keys for all schemes are formulated the same way. 
The public key and master key are formulated based on the output of the $\Setup$ function. 
The public key consists of 3 elements from $\bbG_0$ and a single element from $\bbG_1$. 
The master key consists of a single element from each of $\bbG_0$ and $\bbZ_p$. 

The private key size of a single user is formulated based on the output of the $\KeyGen$ function. 
It consists of $(2|\bbA_j|+1)$ elements from $\bbG_0$, where $|\bbA_j|$ is the number of attributes possessed by the $j^{th}$ data user.

\subsubsection{Ciphertext Size}
Ciphertext size represents the bit-length of the ciphertext. Ciphertext size is based on the output of the $\Encrypt$ function for each scheme.

For CP-ABE, the total size of all generated ciphertexts from all levels consists of $(2(|X_1|+\dots+|X_k|)+k)$ elements from $\bbG_0$ and $k$ elements from $\bbG_1$. 
Using this scheme, the lower levels must accommodate attributes of the higher levels. 
Ciphertext size can potentially end up large in size due to attribute replication at each level.

The single ciphertext generated by FH-CP-ABE~\cite{wang2016efficient} consists of $(2|X|+k)$ elements from $\bbG_0$ and $(v|\bbA_T|+k)$ elements from $\bbG_1$. In this scheme, the ciphertext size depends on $|X|$, $|\bbA_T|$, $v$ and $k$. As the size of these sets grow, the ciphertext size can grow exponentially based on how the tree $\calT$ is constructed.

P-MOD generates ciphertexts in a similar approach to those generated by CP-ABE. 
The total size of all generated ciphertexts consists of $(2[|Y_1|+\dots+|Y_k|]+k)$ elements from $\bbG_0$ and $k$ elements from $\bbG_1$. 
However, the size of the ciphertext generated by P-MOD is shown to be smaller in size than CP-ABE in all instances.
This is based on the composition of P-MOD's access structure which does not duplicate attributes, therefore generating smaller ciphertexts.
In conclusion, P-MOD minimizes the number of attributes in each leveled access tree and therefore minimizes the size of the ciphertexts.

\section{Simulations}
In this section, the results of various simulations are presented to support the performance analysis discussed in Section \ref{sec:performance}. 
P-MOD is implemented and simulated in Java using the CP-ABE toolkit \cite{cpabekit} and the Java Pairing-Based Cryptography library (JPBC) \cite{de2011jpbc}. 
For comparison, simulations are also conducted for CP-ABE~\cite{bethencourt2007ciphertext} and FH-CP-ABE~\cite{wang2016efficient} under the same conditions as P-MOD. 
All simulations are conducted on an Intel(R) Core(TM) i5-4200M at 2.50 GHz and 4.00 GB RAM machine running the Windows 10 operating system. 

In the simulations, the number of levels $k$ is equivalent to the number of file partitions being shared. 
The total number of different user attributes applied to users across all levels within $H$ is represented as $N$.
Both variables are compared for key generation time-cost, encryption time-cost, and decryption time-cost.
The experiments include applying the values for $k = \{ 3,6,9 \}$, to each value of  $N = \{ 10,100 \}$, resulting in 6 experimental cases. 
Within each case, it is possible to distribute the user attributes among the levels in numerous ways. 
For example, when testing the schemes with a hierarchy $H$ for $k=3$ and $N=10$, there are 36 different ways to distribute the 10 user attributes among the 3 levels, excluding the cases of having zero user attributes at any level. 
A possible way of distributing them could be $\bbA_{\calL_1}=3$, $\bbA_{\calL_2}=3$ and $\bbA_{\calL_3}=4$, where $\bbA_{\calL_1}$ represents the set of user attributes possessed by users in the highest level within $H$ and $\bbA_{\calL_3}$  the set of user attributes possessed by users in the lowest level. 
As the values of $k$ and $N$ increase, the possible ways of user attribute distribution among all levels increases. 
Without loss of generality, we assume that the number of attributes for each level follows the normal distribution in our simulations.

The data set used in the simulations is sampled from the Census Income data set \cite{Lichman:2013}. 
The sample consists of 30,163 records, and each record is composed of 9 different record attributes. 
The records are stored in a Microsoft Excel file $\calF$, with a total size of 2.42MB. 
For simulation purposes, it is assumed that data sensitivity is defined over record attributes (vertical columns), corresponding to the second approach described in Section \ref{subsec:file}. 

$\calF$ is partitioned into $k$ sections, where $\forall F_i \in \calF$ and $1 \leq i \leq k$, each $F_i$ represents one or more full columns of $\calF$. Table \ref{tab:partitions} represents the 9 record attributes distribution of $\calF$ into each partition $F_i$, based on the given value for $k$ in each scenario.
\begin{table}[h]
\centering
\caption{Attribute distribution in each partition $F_i$} 
\label{tab:partitions}
 \begin{tabular}{|c|c|c|c|c|c|c|c|c|c|}  \hline
 $\textbf{k}$ & $F_1$ & $F_2$ & $F_3$ & $F_4$ & $F_5$ & $F_6$ & $F_7$ & $F_8$ & $F_9$\\ \hline
 \textbf{3} & 2 & 3 & 4 & - & - & - & - & - & -\\  \hline
 \textbf{6} & 1 & 1 & 1 & 2 & 2 & 2 & - & - & -\\ \hline  
 \textbf{9} & 1 & 1 & 1 & 1 & 1 & 1 & 1 & 1 & 1\\  \hline
\end{tabular}
\end{table} 

All simulations are performed in consideration of users at the most sensitive (highest) level within the hierarchy. 
This approach considers the most complicated applications, where a user needs to gain access to the entire file. Fig. \ref{Fig:time_graph} summarizes the time expended by the three schemes to generate a private key for the user, encrypt all partitions of $\calF$ and decrypt all partitions respectively, in all 6 scenarios. 
By analyzing the figures, some observations can be derived. The results are summarized in the following sections.

\subsection{Key Generation Time-Cost}

To measure the time to generate a private key for a user, the same attribute and level conditions are applied to all three schemes. P-MOD outperforms CP-ABE and FH-CP-ABE in all experimental evaluations. As illustrated in Fig.~\ref{Fig:time_graph}(a), the time taken to generate a private key for a user at the highest level in CP-ABE and FH-CP-ABE, is independent of the value of $k$ and remains nearly constant when $N$ is kept constant, for both values of $N=\{10,100\}$ tested. 

\begin{figure*}[t]
\centering
\begin{tabular}{ccc}
\includegraphics[width=.32\textwidth]{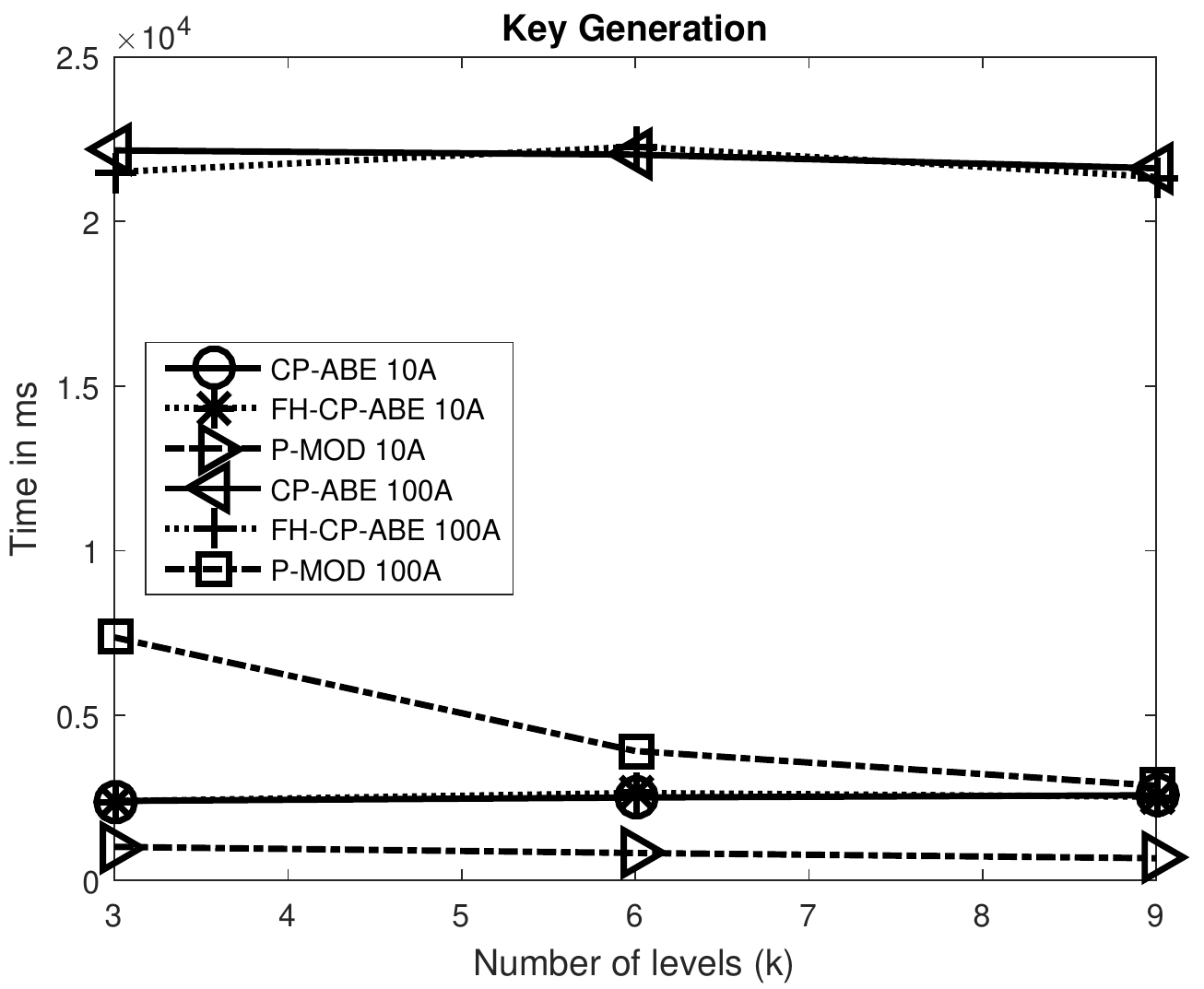}& \includegraphics[width=.32\textwidth]{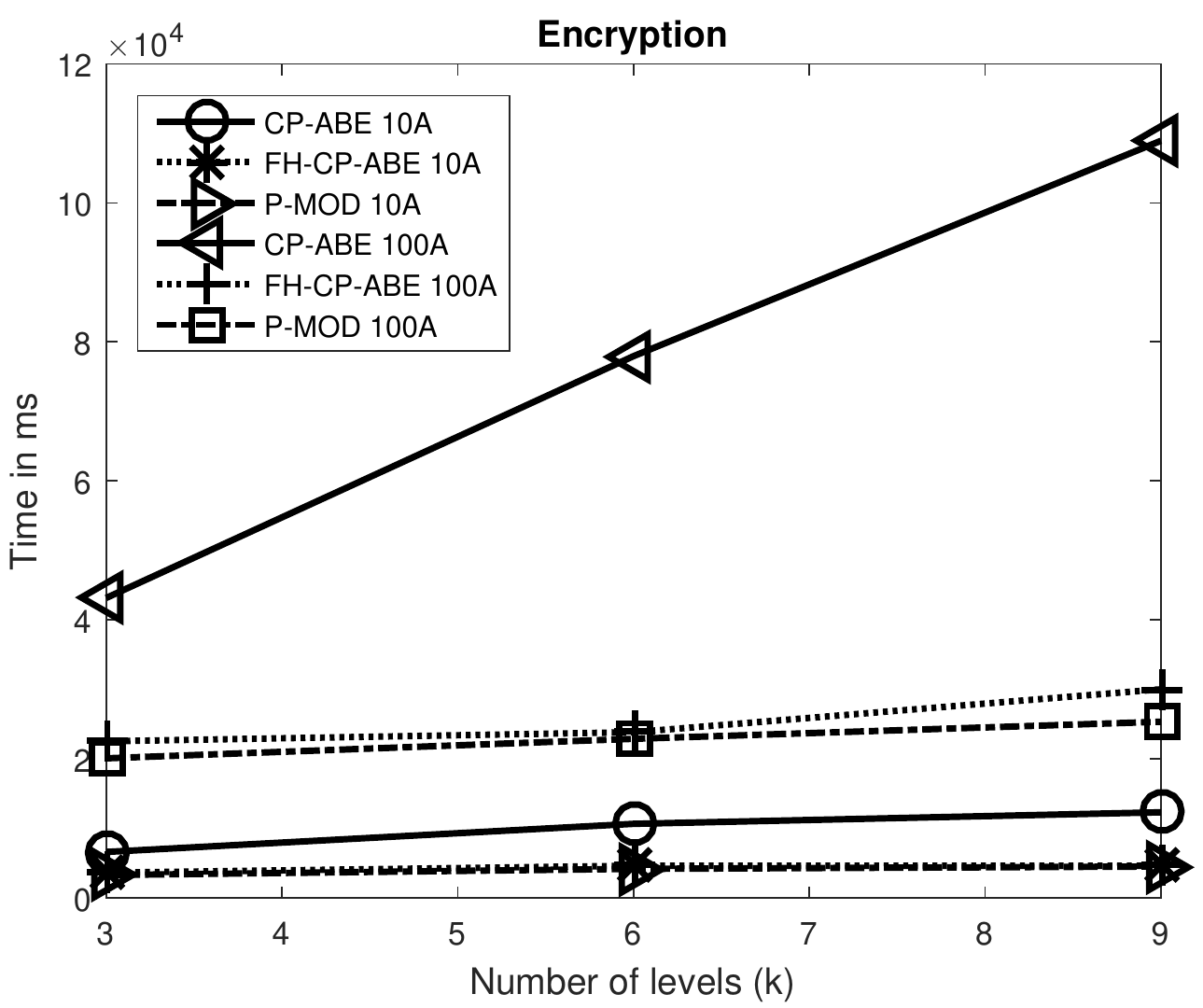}
& 
\includegraphics[width=.32\textwidth]{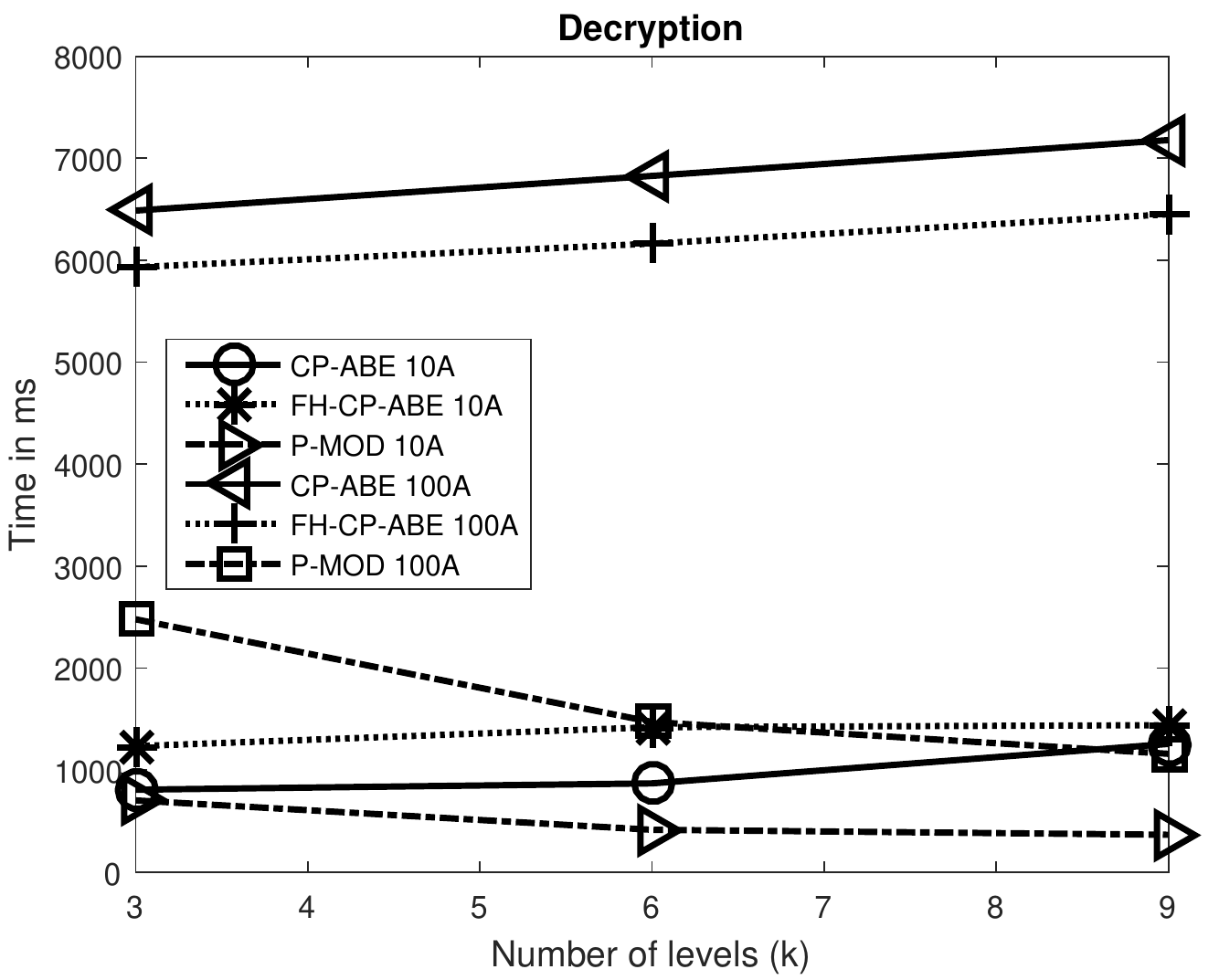}\\
(a) & (b) & (c)
\end{tabular}
\caption{Performance comparison: (a) Key generation time, (b) Encryption time, and (c) Decryption time}
\label{Fig:time_graph}
\end{figure*}

In comparison, P-MOD reacts differently. When the total number of user attributes is normally distributed among the levels, the time taken to generate a private key by P-MOD is approximately the reciprocal of the value of $k$ multiplied by the equivalent time taken by CP-ABE or FH-CP-ABE to perform the same function. For example, the time taken by CP-ABE and FH-CP-ABE in the case where $k=6$ and $N=100$ is approximately 2.2 seconds. The time taken by P-MOD in the same conditions is approximately 0.35 seconds, which is nearly $\frac{1}{6}$ times the time-cost of the other schemes. Therefore, the time taken to generate a private key by P-MOD is inversely proportional to the value of $k$ in the hierarchy. This is true for a constant value N with normally distributed attributes.    

Another observation that can be made from Fig. \ref{Fig:time_graph}(a) is the effect of the values $N$ on the time expended. The time-cost is directly related to the value $N$ and the degree of this proportional increase is different for each scheme as the value $N$ increases. 

We define variable $\delta$ as the difference in time-cost of two experimental evaluations of the same scheme at $N=10$ and $N=100$, while keeping $k$ fixed. In Fig. \ref{Fig:time_graph}(a), consider the simulated values for each scheme at $k=9$. Both CP-ABE and FH-CP-ABE result in large time-cost changes, approximately $\delta_{\mathrm{CP-ABE}} = 19039$ms and $\delta_{\mathrm{FH-CP-ABE}} = 18821$ms, while P-MOD results in a smaller value, $\delta_{\mathrm{P-MOD}} = 2204$ms. Based on these values, $\delta_{\mathrm{P-MOD}}$ is approximately 11.7\% of  $\delta_{\mathrm{CP-ABE}}$ and $\delta_{\mathrm{FH-CP-ABE}}$, resulting in an approximately 88.3\% improvement.  The significance of this observation can be seen in applications where $N$ is a large value. Efficiency can be gained in  time expenditure by utilizing P-MOD, versus CP-ABE and FH-CP-ABE.

\subsection{Encryption Time-Cost}
The encryption time-cost is the time it takes by each scheme to perform the encryption function over all partitions of $\calF$. 
Fig. \ref{Fig:time_graph}(b) represents the time expenditure of each scheme under all six experimental scenarios. 
P-MOD surpasses both CP-ABE and FH-CP-ABE in every experimental case. 
For example, compare the time duration of the three schemes at $k=9$ and $N=100$. 
The time duration for CP-ABE is approximately 4.3 times of P-MOD to perform encryption. 
Similarly, FH-CP-ABE is approximately 1.2 times of P-MOD to perform encryption. 

All schemes follow a direct proportional pattern as the values $k$ and $N$ increase. 
These results prove the correctness of the encryption computational complexity analysis presented in Section \ref{enc_comp_complexity}. 
The encryption function in all schemes involves a number of $f_{G_0}$ and $f_{G_1}$ operations that are dependent on both the values $k$ and $N$. 
However, based on P-MOD's enhanced hierarchical structure, it is able to outperform both CP-ABE and FH-CP-ABE. 
In addition to this, the effect of changing the value $N$ is also illustrated clearly in Fig. \ref{Fig:time_graph}(b). 
For example, when $k=9$, $\delta_{\mathrm{P-MOD}}$ is approximately 21.6\% of $\delta_{\mathrm{CP-ABE}}$ and approximately 82\% of $\delta_{\mathrm{FH-CP-ABE}}$. 

\subsection{Decryption Time-Cost}

As previously discussed, the experiments are performed in the perspective of a user that appears at the highest level of the hierarchy. 
Taking this into account, the decryption time-cost is defined as the time for the user to successfully decrypt all ciphertexts $EF_i$ corresponding to all partitions  $F_i$, if the user possesses the correct set of user attributes. 
Fig. \ref{Fig:time_graph}(c) illustrates the time to perform the decryption function by each scheme.  
The decryption function of both CP-ABE and FH-CP-ABE both involve $e$ and $f_{G_1}$ operations that are dependent on the values $k$ and $N$. 
As these values increase, the decryption time-cost increases linearly for both schemes, proving the correctness of the decryption complexity analysis in Section \ref{dec_comp_complexity}. 

When measuring the decryption time-cost, P-MOD outperforms both CP-ABE and FH-CP-ABE. This is due to the time-cost being inversely proportional to the value of $k$. The time-cost decreases as the value of $k$ increases while $N$ is kept constant. 

P-MOD's decryption time-cost does not severely increase while the value $N$ changes from 10 to 100.
In contrast to this, CP-ABE and FH-CP-ABE are greatly affected, as seen in Fig. \ref{Fig:time_graph}(c). 
For example, when $k=9$, $\delta_{\mathrm{P-MOD}}$ is approximately 13.4\% of $\delta_{\mathrm{CP-ABE}}$ and approximately 15.8\% of $\delta_{\mathrm{FH-CP-ABE}}$. 
The decryption time-cost of P-MOD is expected to drop when $k$ increases while keeping the file size constant.

In summary, for a hierarchical organization with many levels, the simulation results show that P-MOD is significantly more efficient at generating keys, encryption, and decryption than that of both CP-ABE and FH-CP-ABE schemes. 

\section{Conclusion \label{sec:conclusion}}
The numerous benefits provided by the cloud have driven many large multilevel organizations to store and share their data on it. 
This paper begins by pointing out major security concerns data owners have when sharing their data on the cloud.  
Next, the most widely implemented and researched data sharing schemes are briefly discussed revealing points of weakness in each. 
To address the concerns, this paper proposes a Privilege-based Multilevel Organizational Data-sharing scheme~(P-MOD) that allows data to be shared efficiently and securely on the cloud. 
P-MOD partitions a data file into multiple segments based on user privileges and data sensitivity. 
Each segment of the data file is then shared depending on data user privileges.
We formally prove that P-MOD is secure against adaptively chosen plaintext attack assuming the DBDH assumption holds. 
Our comprehensive performance comparison with the two most representative schemes shows that P-MOD can significantly reduce the computational complexity while minimizing the storage space.

\bibliographystyle{IEEETran}
\bibliography{P-MOD}

\end{document}